\documentclass[10pt]{article}

\usepackage{graphicx}

\usepackage{a4wide}  

\usepackage{url}  
\usepackage[utf8]{inputenc}  
\usepackage{subfigure}  
\usepackage{latexsym}  
\usepackage{amsmath} 
\usepackage{amssymb} 
\newtheorem{lemma}{Lemma}  
  
\newtheorem{corollary}{Corollary}

\newcommand{\pr}{\textrm{pref}}
\newcommand{\ov}{\textrm{ov}}
\newcommand{\pe}{\textrm{period}}
\newcommand{\cn}{\textrm{\it sp}}
\newcommand{\opt}{\mbox{\it OPT}}
\newcommand{\opts}{\mbox{\it OPT}_{\sigma}}
\newcommand{\be}{\beta}

\newtheorem{theorem}{Theorem}  

\newenvironment{proof}{\noindent {\it Proof.}}{$\Box$\vskip1ex}

\graphicspath{{figures/}}

\begin{document}

\title{A note on the shortest common superstring\\ of NGS reads}

\author{Tristan Braquelaire\thanks{LaBRI and CBiB, University of Bordeaux, France.} \and Marie Gasparoux$\,^*$ \and Mathieu Raffinot\thanks{CNRS, LaBRI and CBiB, University of Bordeaux, France.} \and Raluca Uricaru$\,^*$}

\maketitle

\begin{abstract}
The Shortest Superstring Problem (SSP) consists, for a set of strings
$S = \{s_1,\cdots,s_n\}$, to find a minimum length string that
contains all $s_i, 1\leq i \leq k$, as substrings.

 This problem is proved to be {\em NP-Complete} and APX-hard. Guaranteed
 approximation algorithms have been proposed, the current best ratio
 being $2\frac{11}{23}$, which has been achieved following a long and
 difficult quest. However, SSP is highly used in practice on next
 generation sequencing (NGS) data, which plays an increasingly important role in sequencing. In this note, we show that the SSP approximation ratio can be
 improved on NGS reads by assuming specific characteristics of NGS data
 that are experimentally verified on a very large sampling set.
\end{abstract}

\section{Introduction}

The Shortest Superstring Problem (SSP) consists, for a set of strings
$S = \{s_1,\cdots,s_n\}$, in constructing a string $s$ such that any
element of $S$ is a substring of $s$ and $s$ is of minimal length.
For an arbitrary number of sequences $n$, the problem is known to be {\em NP-Complete}~\cite{GALLANT198050, Garey1990}
and APX-hard~\cite{Blum:1994}. Lower bounds for the achievable approximation ratios
on a binary alphabet have been given by
Ott~\cite{Ott1999}. The best known approximation ratio so far is $2
\frac{11}{23} \approx  2.478$ \cite{Mucha13} after a long series of improvements
\cite{l-tdstls-90, Blum:1994,KPS94,
  Armen1995,Armen199829,Breslauer1997340,
  Czumaj199774,Sweedyk:1999,TengY97, KaplanS05, PaluchEZ12}.

In the meantime, SSP compression algorithms have been designed as
sub-routines of the previous ones. The idea is to ensure a fixed
compression ratio between the sum of the lengths of the sequences of
the set and the optimal superstring on this set. The greedy algorithm
is such a compression algorithm that is proven to achieve a
compression ratio of at least $\frac{1}{2}$, while the best
compression algorithm achieves a ratio of $\frac{38}{63}$
\cite{KPS94}.

In this note, we focus on practical applications of SSP, like assembling biological sequences, mostly DNA sequences with an alphabet of $\{A,C,G,T\}$ named {\em bases}, but also on proteome
sequences with a 26 letter alphabet corresponding to {\em amino acids}. SSP is
used in contig reconstruction step, contigs that subsequently need to be organised.

Over the past decade, the landscape of sequencing and assembly deeply
changed, with the increasing development of Next Generation Sequencing (NGS)
devices. These relatively cheap devices produce, from a ``soup'' of cells, millions of randomly read, short, equal length DNA sequences in a single {\em run}.
Each sequence is typically 32 to 1000 bases long, with a small and still decreasing cost per
base. Such sequences are named {\em reads}. NGS technology allows to
tackle new challenges in biology and medecine; the exponential
increase of sequencing demands leads to the creation of more and
more sequencing platforms, dealing with NGS data at 99\%.

Considering the specificity of read sequences, is it possible to propose better
approximation algorithms for this type of data? This research,
similar to the one targeting better algorithms for small-world graphs in
social networks, aims to better suit the actual data.

This note is a first step in this direction. We first model the read
sequences more finely thus, according to our examples,
better matching the experimental data.  Then, we derive a better
approximation ratio algorithm by using the properties of the
reads. For instance, on the set SRR069579, we reach a $2.0738$
approximation ratio (see Table \ref{SRR069579}). To our knowledge, the
only related work is \cite{GolovnevKM13}, where the sequences have the
same length. Up to $7$ bases, they propose a better approximation
ratio based on De Bruijn graphs. However, these sequences are way
shorter than real-world reads.

Note that some theoretical variations of SSP have also been studied
\cite{Yu16a, CrochemoreCIKRRW10}. Here we do not dwell on these
studies since their focus is far from ours, neither do we detail the
greedy algorithm approximation conjecture, which is a subject by
itself \cite{TARHIO1988131,KaplanS05,FiciKRRW16}.

\section{Modeling of reads}

NGS reads have some specific properties that we model and exhibit on
real sets of reads.

For a string $s$ of length $n$, any integer $1 \le p \le m$ is a
\emph{period} of $s$ if $s[i] = s[i+p]$ for all $1 \le i \le
m-p$. Note that $s$ always has at least one period, corresponding to its
length. The smallest period of $s$ is called \emph{the period} of $s$,
and denoted $\mbox{\em period}(s)$.

We consider SSP on $n$ reads $S=\{s_1,s_2,\ldots ,s_n\}$ of length
$m>0$, where $m\ll n$. We now consider the period of each read. We
denote $n(i)$, $1 \leq i \leq m$, the number of reads of period $i$. 


Let $0 \leq \alpha \leq 1$ be a parameter and let $\cn =
\sum_{i=1}^{m\alpha} \frac{n(i)}{i}$. We express $\cn$ (for {\em s}mall {\em p}eriod) as a percentage
of $\frac{n}{m}$ relatively to the value of $\alpha$ and we
denote  $\cn = \mbox{\em perc}_{\alpha} \frac{n}{m}.$\\

A strong characteristic of a set of reads is that even for ratios $ 0.8
< \alpha < 1 $, $\cn$ is very small compared to $n$. The order of
magnitude is $\cn$ being a few per hundred of $\frac{n}{m}$. For a
large panel of sets of reads on which we tested our approach, we found
for $ 0.8 < \alpha < 1 $ a $\cn$ value inferior to $0.02 \, n/m$. In
Figure \ref{plotsreads}, we show such four sets of reads with lengths
of 32, 36, 98 and 200.




\begin{figure}[t]
\begin{centering}
\includegraphics[width=5cm]{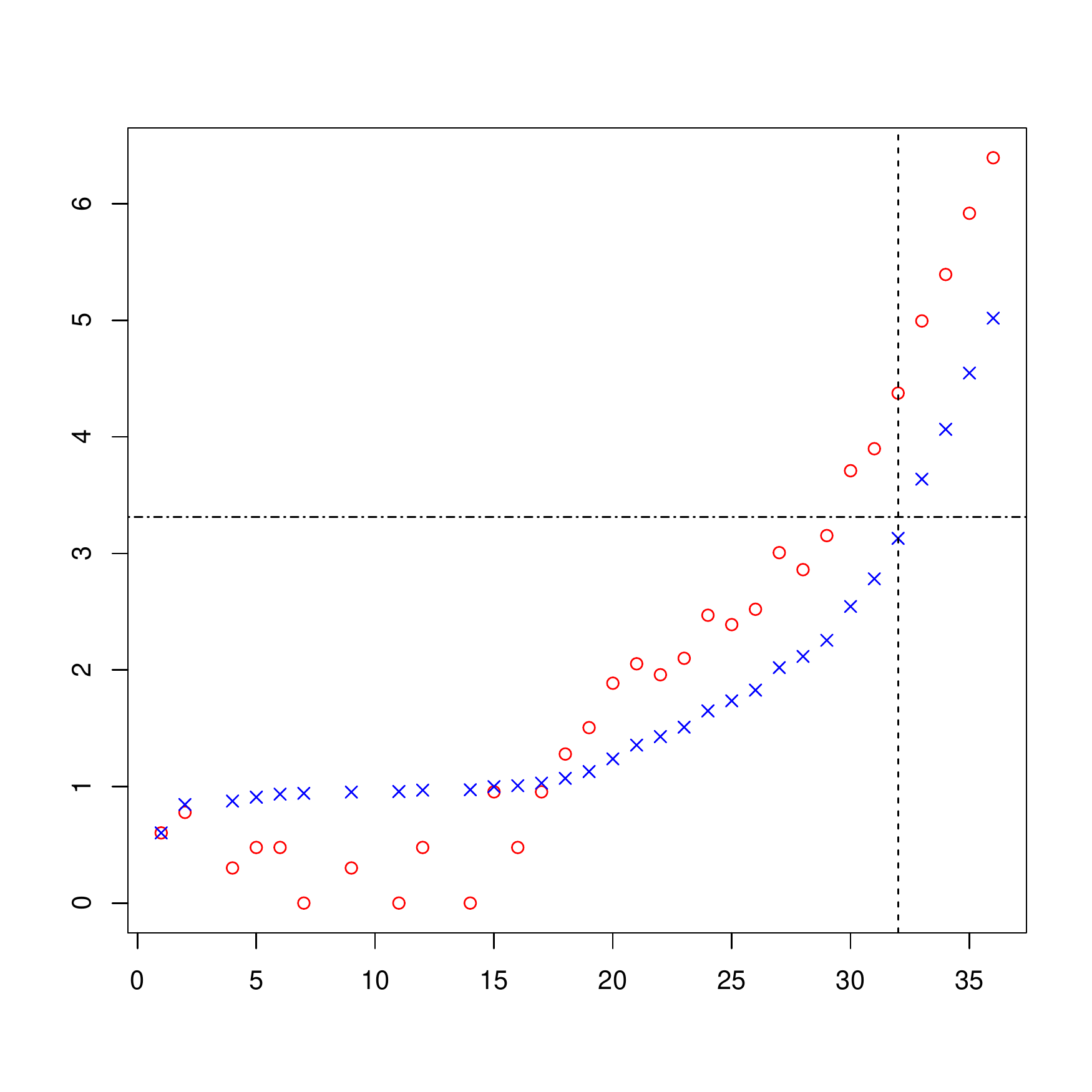}
\includegraphics[width=5cm]{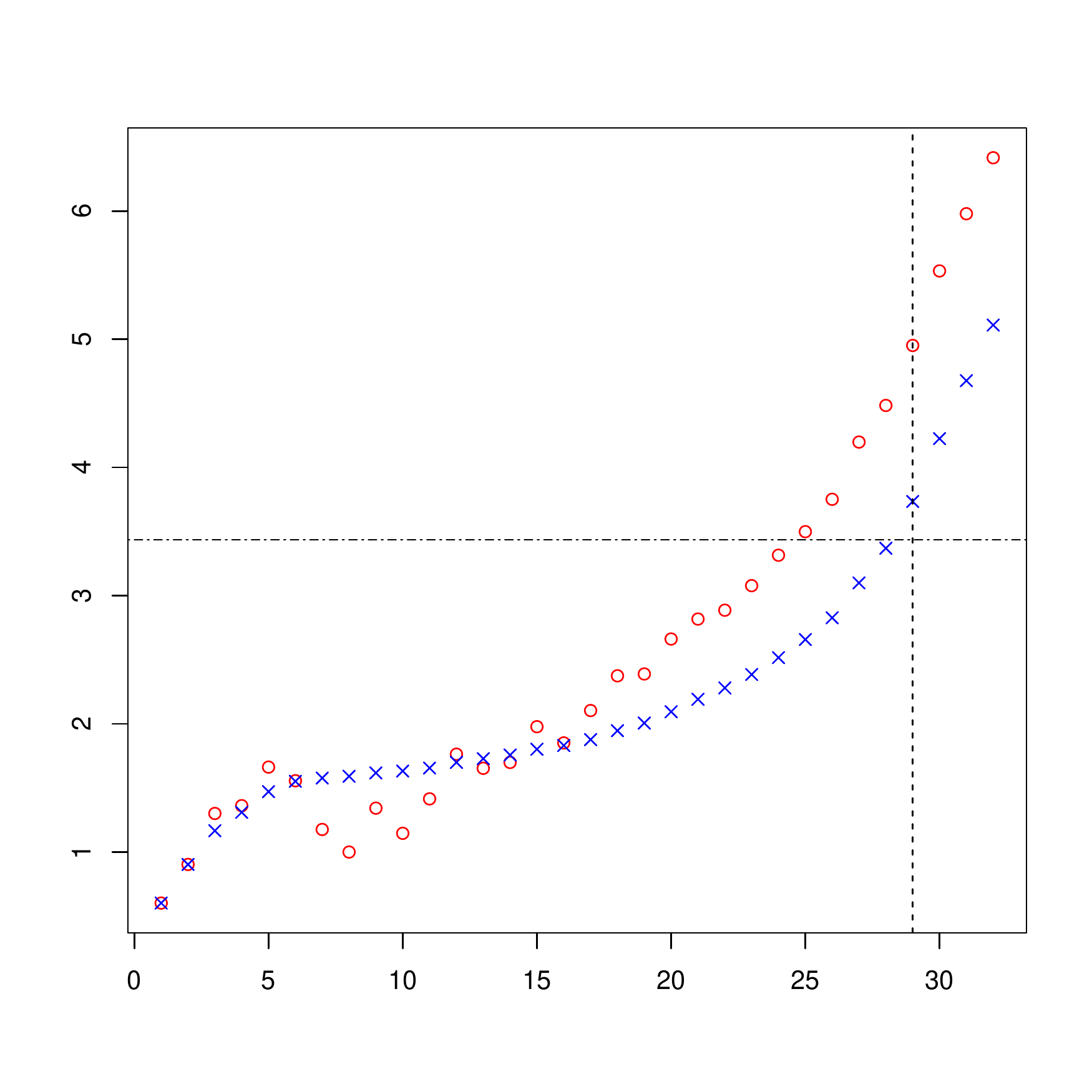}\\[-4mm]
\includegraphics[width=5cm]{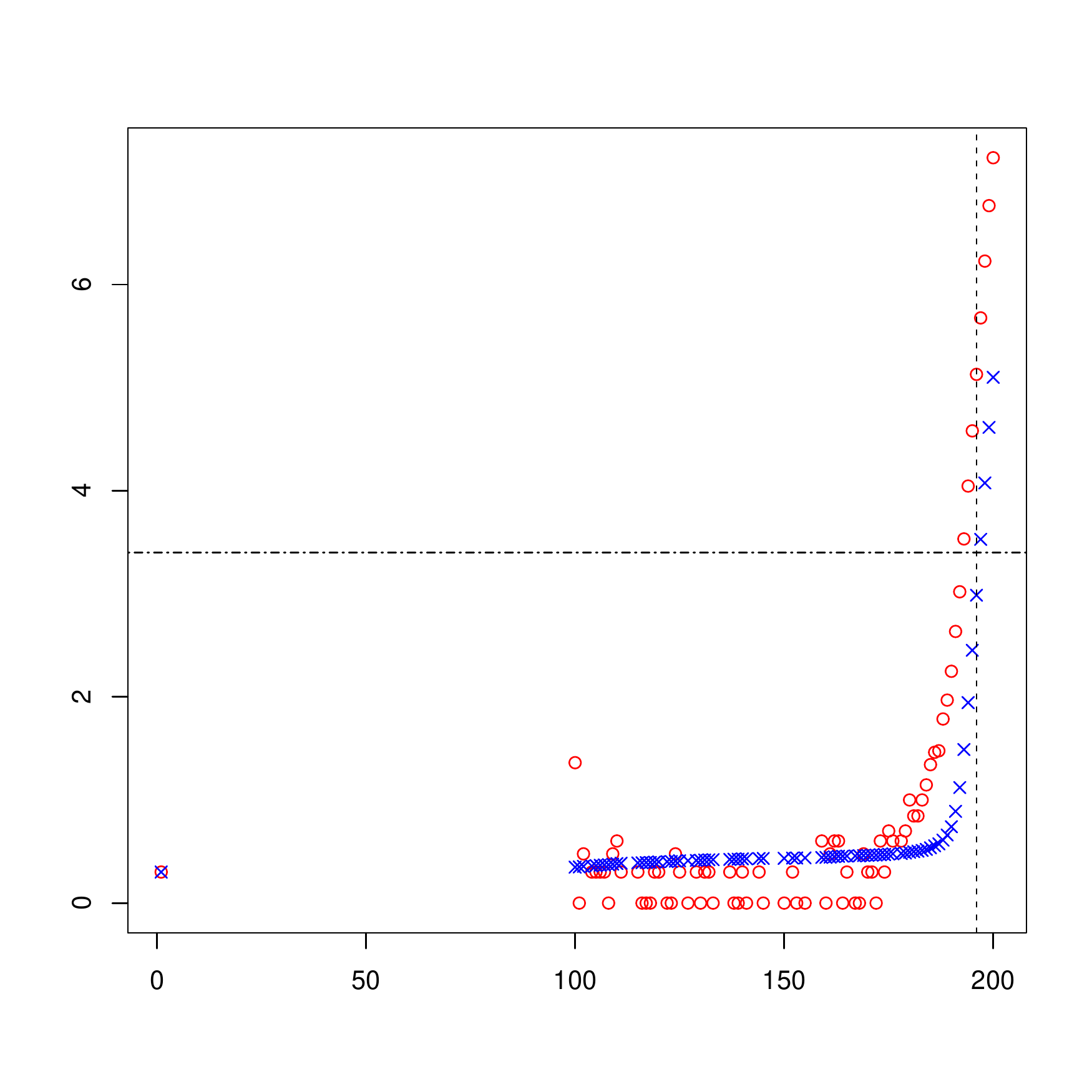}
\includegraphics[width=5cm]{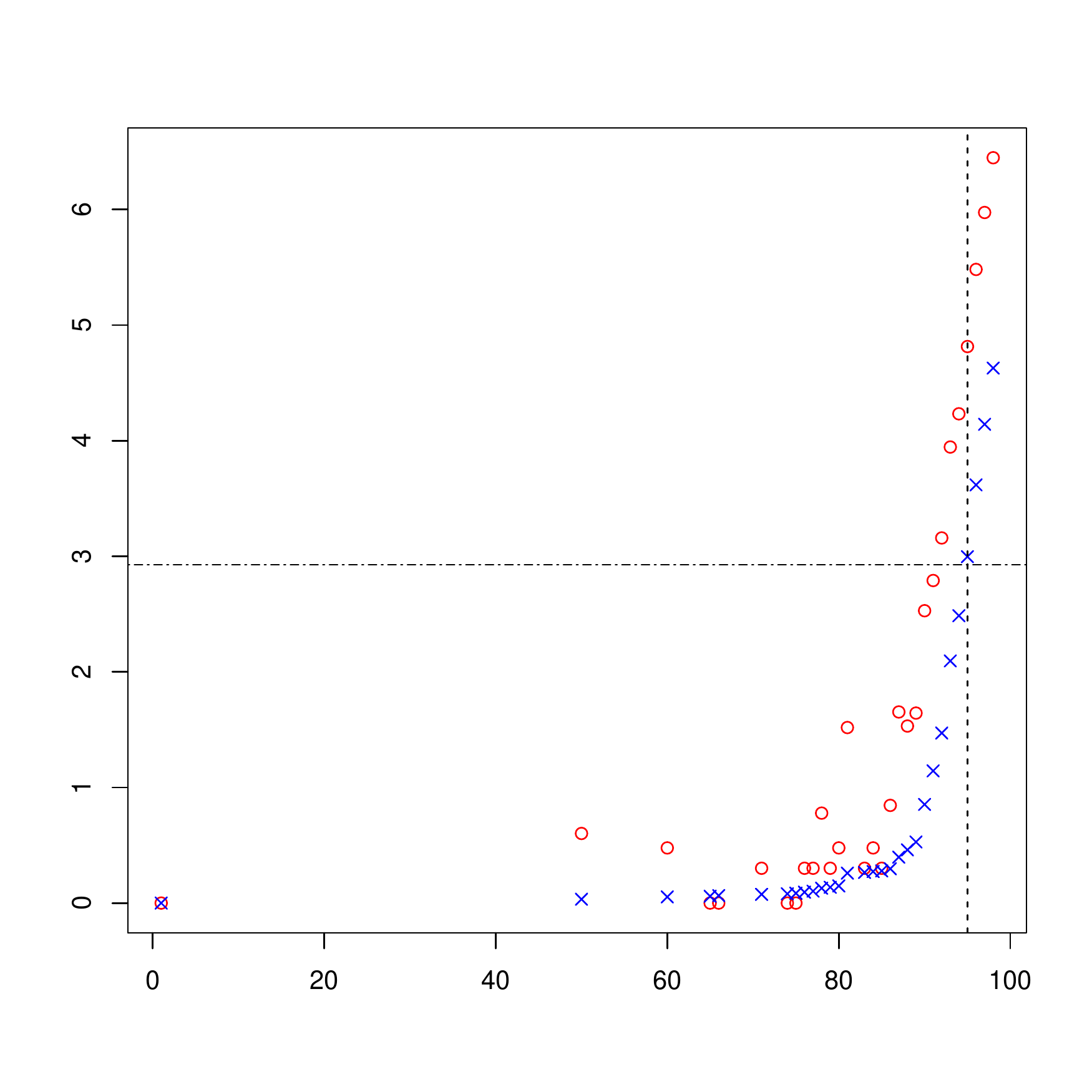}
\vspace*{-0.3cm}
\caption{Sets of reads from left to right, top to bottom: SRR069579
  (human), ERR000009 (yeast), SRR211279 (human), SRR959239
  (human). The $x$-axis is the period and the $y$-axis is in $\log_{10}$
  scale. The circles represent $n(x)$ and the crosses $\sum_{i=1}^{x} \frac{n(x)}{x}.$ The dash
  vertical line corresponds to the final $m\alpha$ (computed in the
  experimental results in Section \ref{exprest}) and the horizontal to
  $0.02\frac{n}{m}.$}
\label{plotsreads}
\end{centering}
\end{figure}

\section{Approximation algorithm}

For two strings $u,v$ we define the \emph{overlap} of $u$ and $v$,
denoted $\ov(u,v)$, as the longest suffix of $u$ that is also a prefix
of $v$. Also, we define the \emph{prefix} of $u$ relatively to $v$, denoted
$\pr(u,v)$, as the string $x$ such that $u=x\,\ov(u,v)$, {\em i.e.}, the
prefix of $u$ that does not overlap $v$.

The \emph{prefix graph} (also called the \emph{distance graph}) of $S$
is a complete directed graph with the vertex set $S$ and the edges
$(s_i,s_j)$ of weight equal to the length $|\pr(s_i,s_j)|$.


We consider the classical algorithm of
\cite{Blum:1994,vazirani}, which gives a general framework. This algorithm is proved to be a 3 approximation algorithm in the general case. We prove below that applied on NGS data the
approximation factor can be improved. The scheme of the algorithm is the following:

\begin{enumerate}
\item Compute a maximal cycle decomposition on the prefix graph
\item For each cycle $c_i$ choose one of the strings in $c_i$ as a representative string $r_i$.
\item $\sigma_i= (\pr(r_i,r_{i+1}) \cdot ... \cdot \pr(r_k,r_{i})) \cdot r_i$ (cycle from $r_i$ concatenated with $r_i$).
\item Let  $S_{\sigma} = \{ \sigma_i \}$ and $w_{\sigma}$ as a concatenation of all $\sigma_i$.
\item Compress $w_{\sigma}$ using an SSP compression algorithm.
\end{enumerate}

The cycle decomposition produces cycles of several lengths. The {\em
  period} of a cycle is given by its length. We split the set of cycles, in
two parts, the {\em small} cycles of period less than or equal to
$m\alpha$, and the larger ones, denoted {\em large}. We now focus on
the number of {\em small} cycles. The weight of a cycle is the sum of
the weights of its edges and let $\mbox{wt}(\mathcal{C})$ be the sum of
the weights of all cycles.

\begin{lemma}
\label{cgreaters}
Let $c \in \mathcal{C}$ be a cycle and $s$ a sequence in the
cycle,  then $\mbox{period}(c)\geq \mbox{period}(s)$
\end{lemma}
\begin{proof}
Each sequence in the cycle can be expressed by turning around the cycle. If $\mbox{period}(c) < \mbox{period}(s)$, then
$\mbox{period}(c)$ is also a period of $s$, which is smaller than its smallest
period, contradiction.
\end{proof}

\begin{corollary}
\label{corocgreaters}
Let $c \in \mathcal{C}$ be a cycle and $s_1\ldots s_k$ the sequences
in $\mathcal{C}$. Then $\mbox{period}(c)\geq
\mbox{max}\{\mbox{period}(s_i)\}.$
\end{corollary}
\begin{proof}
Directly derives from lemma \ref{cgreaters}.
\end{proof}

\begin{lemma}
\label{cyclesmallperiod}
Let $c \in \mathcal{C}$ be a cycle and $s_1\ldots s_k$ the sequences
in $\mathcal{C}$. If $\mbox{period}(c)\leq m \alpha$, $\mbox{period}(s_i) \leq m \alpha.$
\end{lemma}
\begin{proof}
By corollary \ref{corocgreaters}, the periods of all sequences in a cycle are smaller than or equal to the period of the cycle. 
\end{proof}


\begin{corollary}
\label{nbsmallcycles}
Let $1 \leq i \leq m$, the maximal number of cycles of period less or equal to $i$ is bounded by $\frac{1}{2}\sum_{k=1}^i n(i).$
\end{corollary}
\begin{proof}
A cycle contains at least two sequences. By Lemma
\ref{cyclesmallperiod}, all the sequences in $c$ of period $i$ must have a period
less than or equal to $i$ and there are only $\frac{1}{2}\sum_{k=1}^i n(i)$ such sequences.
\end{proof}


\subsection{Analysis of the algorithm}

We bounded the number of {\em small} cycles relatively to
$\alpha$. Let us now take this into account while analysing the
approximation algorithm. Obviously, $w_{\sigma}$ is a
superstring of $S$. Let us bound its size.


\begin{lemma}
$$|w_{\sigma}|= \sum \left| \sigma_i \right| \leq \mbox{wt}(\mathcal{C})
  + \mbox{wt}(\mathcal{C}) \frac 1 \alpha + \frac {\cn \cdot m} 2 \leq (1+\frac
  1 \alpha) \opt + \frac{\cn \cdot m}{2}$$ 
\end{lemma}
\begin{proof}
A $\sigma_i$ is formed on a cycle as $(\pr(r_i,r_{i+1}) \cdot
... \cdot \pr(r_k,r_{i})) \cdot r_i.$ We first sum over all the
$\sigma_i$ the prefixes of each $\sigma_i$ corresponding to the
cycle. This leads to a first global $\mbox{wt}(\mathcal{C}).$ Then we
consider the sizes of the $r_i$ for the large cycles. The point is that
all $r_i$ have the same length $m$, and that each $r_i$ can be represented (or {\em expressed})
by turning around the cycle it corresponds to (see Figure \ref{expressing}). As {\em large} cycles have a
period of at least $m \alpha$, turning $\frac{1}{\alpha}
\pe(c_i)$ around the cycle $c_i$ is enough to read $r_i$. Thus, the sum over all
   {\em large} cycles of $[r_i|$ is bounded by $\mbox{\em wt}(\mathcal{C})
     \frac 1 \alpha.$ 

\begin{figure}[h]
\centering 
\includegraphics[width=6cm]{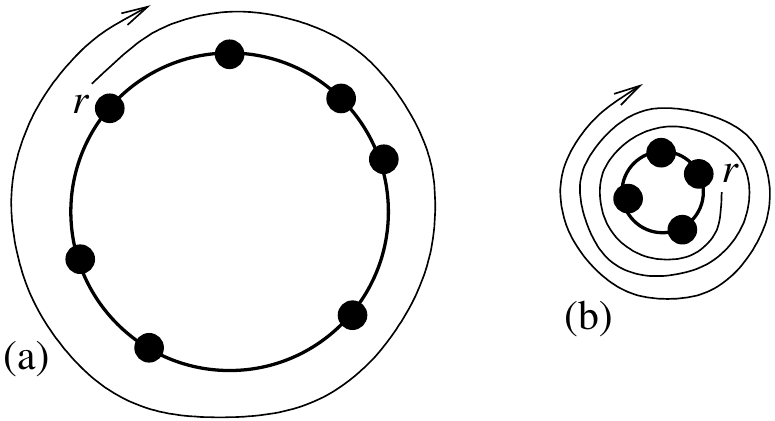}
\caption{Expressing a representative $r$ over the cycle it belongs to. The
  period of cycle (a) is larger than $m\alpha$ and thus the expression
  of $r$ requires only $1/\alpha$ cycles. The period of cycle (b) is $2 \leq i <
  m\alpha$ and the expression of $r$ thus requires $m/i$ cycles, which can be at maximum $m/2$}
 \label{expressing}
\vspace{-0.3cm}
\end{figure}

The remaining step is counting the sum of the $r_i$ corresponding to $n_{\alpha}$ small cycles.

As, by corollary \ref{nbsmallcycles}, there are at most $\cn/2$ such
cycles, the sum of the corresponding $r_i$ is bounded by $ \cn
\cdot m/2$. This would already be an acceptable bound since $\cn$ is small
relatively to $m/n$. But this implies counting $m$ for all {\em
  small} cycles, independently of the periods of the cycles, which can
vary from 2 to $m\alpha$. The larger the period of the cycle, the less
we need to turn on the cycle to read the representative $r_i$.
Thus our worst case for counting the small cycles from period $1$ to
$m\alpha$ is when there is a maximum of smaller cycles at each step
$2\leq k \leq m \alpha$ and, by corollary \ref{nbsmallcycles}, this
maximum from $k-1$ to $k$ can only be increased by $n(k)$. Expressing the
representative of each such $n(k)$ additionnal cycles of period $k$,
requires $n(k)\frac{m}{k}.$ Thus, the expression of the
representatives of all the small cycles is bounded by $\frac 1 2 \sum_{i=1}^{m
  \alpha} \frac{n(i)\cdot m}{i} = \frac{\cn \cdot m}{2}.$ 

Eventually, as   $\mbox{wt}(\mathcal{C}) \leq \opt$, the result follows.
\end{proof}

We then {\em compress} $S_{\sigma}$ using the guaranteed {\em
  compression} algorithm of $\frac{38}{63}$ \cite{KPS94}, similarly to
the classical approaches related to the superstring approximation. We
define $\mbox{OPT}_{\sigma}$ as an optimal minimal superstring on
$S_{\sigma}$ and $\tau$ as the result of the compression algorithm on
$S_{\sigma}.$ The next lemma \cite{Blum:1994,vazirani} allows us to link
$\opts$ and $\opt$.

\begin{lemma} 
$\opts <\opt+\mbox{wt}(\mathcal{C})$
\label{lessopt}
\end{lemma}

\noindent
By applying the compression algorithm on $S_{\sigma}$, we thus derive the following result:

\begin{lemma}
$|\tau|  \leq 2\opt + \frac{38}{63} \left( \frac{1-\alpha}{\alpha} \right) \opt + \frac{38}{126} \cn \cdot m$
\end{lemma}
\begin{proof}
Lemma \ref{lessopt} gives $\opts <\opt+\mbox{wt}(\mathcal{C}) \leq 2
\opt$ (see Figure \ref{compression}). The distance from $\opts$ to $|w_{\sigma}|$ is greater than or
equal to $|w_{\sigma}|-{2\opt}.$ In the worst case it is equal, then
the compression algorithm applies a compression factor $38/63$ to this
distance, which leads to the result.

\begin{figure}[h]
\centering 
\includegraphics[width=6cm]{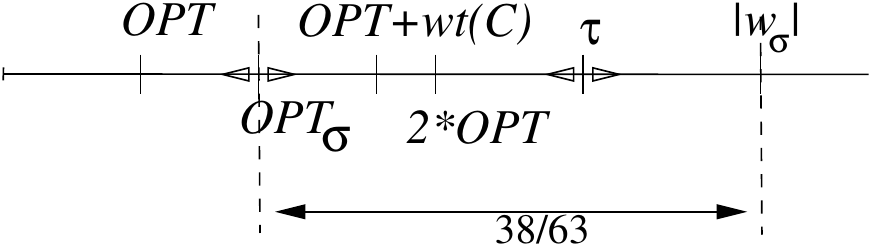}
\caption{Compressing $w_{\sigma}$ using the $38/63$ algorithm \cite{KPS94}}
 \label{compression}
\vspace{-0.3cm}
\end{figure}
\end{proof}

An important point is that $OPT \geq n$, since any superstring contains at least one base of each sequence. As $\cn = \mbox{\em perc}_{\alpha} \frac{n}{m}$, $\frac{38}{63} \frac{\cn m}{2} \leq \frac{ 38}{126} \mbox{\em perc}_{\alpha} \mbox{OPT}$

\begin{theorem}
\label{lastequa}
$$|\tau|  \leq 2\opt + \frac{38}{63} \left( \frac{1-\alpha}{\alpha} \right) \opt + \frac{38}{126} \mbox{\em perc}_{\alpha} \mbox{OPT}$$
\end{theorem}

\begin{table}[!htbp]
{\small \begin{tabular}{|rrrrrrrr|}
\hline
period & nbseq &  cum. nbseq & $\alpha$ & $1+\frac{1}{\alpha}$ &    $2 + \frac{38}{63} \left( \frac{1-\alpha}{\alpha} \right)$ &  $\frac{38}{126} \cn \frac{m}{n}$ &  $\be$\\ \hline
1 & 4 & 4 & 0.0277778 & 37 & 23.1111 & 1.74749e-05 & 23.1111\\
2 & 6 & 10 & 0.0555556 & 19 & 12.254 & 3.0581e-05 & 12.254\\
$\ldots$ & $\ldots$ & $\ldots$ & $\ldots$ & $\ldots$ & $\ldots$ & $\ldots$  & $\ldots$ \\ 
32 & 23746 & 41326 & 0.888889 & 2.125 & 2.0754 & 0.00588868 & 2.08129\\
{\bf 33} & {\bf 98795} & {\bf 140121} & {\bf 0.916667} & {\bf 2.09091} & {\bf 2.05483} & {\bf 0.0189677} & {\bf 2.0738}\\
34 & 247451 & 387572 & 0.944444 & 2.05882 & 2.03548 & 0.0507631 & 2.08624\\
35 & 829535 & 1217107 & 0.972222 & 2.02857 & 2.01723 & 0.154306 & 2.17154\\
36 & 2485202 & 3702309 & 1 & 2 & 2 & 0.455893 & 2.45589\\
\hline
\end{tabular}
}

\caption{SRR069579 read set, 3702309 reads of size 36. $\be = 2 + \frac{38}{63} \left( \frac{1-\alpha}{\alpha} \right) + \frac{38}{126} \cn \frac{m}{n} $ }
\label{SRR069579}
\end{table}

\begin{table}[!htbp]
{\small \begin{tabular}{|rrrrrrrr|}
\hline
period & nbseq &  cum. nbseq & $\alpha$ & $1+\frac{1}{\alpha}$ &    $2 + \frac{38}{63} \left( \frac{1-\alpha}{\alpha} \right)$ &  $\frac{38}{126} \cn \frac{m}{n}$ &  $\be$\\ \hline

1&	4&	4&	0.03125&	33&	20.6984&	8.83862e-06&	20.6984\\
2&	8&	12&	0.0625&	17&	11.0476&	1.76772e-05&	11.0476\\

$\ldots$ & $\ldots$ & $\ldots$ & $\ldots$ & $\ldots$ & $\ldots$ & $\ldots$  & $\ldots$ \\ 

28&	30474&	61366&	0.875&	2.14286&	2.08617&	0.00518227&	2.09135\\
{\bf 29}&	{\bf 89474}&{\bf	150840}&{\bf	0.90625}&{\bf	2.10345}&{\bf	2.0624}&	{\bf 0.0119997}&{\bf	2.0744}\\
30&	341160&	492000&	0.9375&	2.06667&	2.04021&	0.0371279&	2.07734\\
31&	953389&	1445389&	0.96875&	2.03226&	2.01946&	0.105085	&2.12454\\
32&	2606944&	4052333&	1&	2&	2&	0.285099&	2.2851\\

\hline
\end{tabular}
}

\caption{ERR000009 read set, 4052333 reads of size 32}
\label{ERR000009}
\end{table}

\begin{table}[!htbp]
{\small \begin{tabular}{|rrrrrrrr|}
\hline
period & nbseq &  cum. nbseq & $\alpha$ & $1+\frac{1}{\alpha}$ &    $2 + \frac{38}{63} \left( \frac{1-\alpha}{\alpha} \right)$ &  $\frac{38}{126} \cn \frac{m}{n}$ &  $\be$\\ \hline
1 &	2 &	2 &	0.005	  & 201 &	122.032 & 4.80545e-06 &	122.032\\
100 &	23 &	 25 &	0.5&	3 &	2.60317 &	5.35808e-06 &	2.60318\\
$\ldots$ & $\ldots$ & $\ldots$ & $\ldots$ & $\ldots$ & $\ldots$ & $\ldots$  & $\ldots$ \\ 
195 &	38013&	54574&	0.975&	2.02564&	2.01547&	0.00067972&	2.01615\\
{\bf 196}&	{\bf 134284}&	{\bf 188858}&	{\bf0.98	}& {\bf 2.02041}&	{\bf 2.01231} &	{\bf 0.00232588} &	{\bf 2.01464}\\
197 &	473686 &	662544 &	0.985 &	2.01523 &	2.00919 &	0.00810323&	2.01729\\
198 &	1685038 &	2347582 &	0.99  &	2.0101 & 	2.00609 &	0.0285511 &	2.03464\\
199 &	5811666 &	8159248 &	0.995 &	2.00503 &	2.00303 &	0.0987212 &	2.10175\\
200 &	16944518 &	25103766 &	 1 & 	       2 &	      2 &	0.302286&   2.30229\\
\hline
\end{tabular}
}

\caption{SRR211279 read set, 25103766 reads of size 200}
\label{SRR211279}
\end{table}

\begin{table}[!htbp]
{\small \begin{tabular}{|rrrrrrrr|}
\hline
period & nbseq &  cum. nbseq & $\alpha$ & $1+\frac{1}{\alpha}$ &    $2 + \frac{38}{63} \left( \frac{1-\alpha}{\alpha} \right)$ &  $\frac{38}{126} \cn \frac{m}{n}$ &  $\be$\\ \hline

1&	1&	1&	0.0102041&	99&	60.5079&	7.13344e-06&	60.5079\\
50&	4&	5&	0.510204&	2.96	&2.57905&	7.70411e-06&	2.57906\\

$\ldots$ & $\ldots$ & $\ldots$ & $\ldots$ & $\ldots$ & $\ldots$ & $\ldots$  & $\ldots$ \\ 

94&	17083&	28491&	0.959184&	2.04255&	2.02567&	0.00218336&	2.02785\\
{\bf 95} &	{\bf 65228} &{\bf 93719} &{\bf 0.969388} &	{\bf 2.03158}&	{\bf 2.01905} &	{\bf 0.00708125} &	{\bf 2.02613}\\
96&	302973&	396692&	0.979592	&2.02083&	2.01257&	0.0295941&	2.04216\\
97&	942267&	1338959&	0.989796&	2.01031&	2.00622&	0.098889	&2.10511\\
98&	2804284&	4143243&	1&	2&	2&	0.303013&	2.30301\\

\hline
\end{tabular}
}

\caption{SRR959239 read set,  4143243 reads of size 98}
\label{SRR959239}
\end{table}

\section{Experimental results}
\label{exprest}

We present experimental results for the sets of reads SRR069579 (Table
\ref{SRR069579}), ERR000009 (Table \ref{ERR000009}), SRR211279 (Table
\ref{SRR211279}), and SRR959239 (Table \ref{SRR959239}). 

In each table, for each period $i$ from $1$ to $m$ we show : (a)
$n(i)$, (b) the cumulative number of sequences, (c) the value of
$\alpha$ corresponding to $i/m$, (d) the value of $1+\frac{1}{\alpha}$, (e) $2
+ \frac{38}{63} \left( \frac{1-\alpha}{\alpha} \right)$ which
corresponds to the term of equation \ref{lastequa} due to the large
cycles, (f) $\frac{38}{126} \cn \frac{m}{n}$ which is the part of the
final ratio brought by the small cycles, and eventually (g) $\be  = 2 + \frac{38}{63} \left( \frac{1-\alpha}{\alpha} \right) + \frac{38}{126} \cn \frac{m}{n}$, the final ratio that can be reached by using the value of $\alpha$ from the
previous line in the table.

The resulting approximation ratios on the read sets cited above are
respectively 2.0738, 2.09, 2.01464 and 2.02623.





{\small \bibliographystyle{plain}

}
\end{document}